\documentclass[a4paper,11pt]{article}
\usepackage{amsmath,amssymb,amsthm,graphicx,hyperref,geometry}
\newtheorem{theorem}{Theorem}
\newtheorem{lemma}{Lemma}
\newtheorem{conjecture}{Conjecture}
\newtheorem{proposition}{Proposition}
\newtheorem{corollary}{Corollary}
\newcommand{\md}{\mathcal{M}}
\newcommand{\etaeff}{\eta_{\mathrm{eff}}}

\title{The exact price of local realism in CHSH experiments:\\
a measurement-dependence--detection trade-off surface,\\
a moir\'e phase-locking mechanism that saturates it,\\
and an unmeasured fringe in the fourfold coincidence sum}
\author{Aaron Alai\\ \small Baltimore, Maryland, USA}
\date{August 20, 2026; revised November 25, 2026 (v8)}

\usepackage{url}
\expandafter\def\expandafter\UrlBreaks\expandafter{\UrlBreaks
  \do\/\do\-\do\_\do\.}
\begin{document}
\maketitle

\begin{abstract}
For local hidden-variable accounts of the CHSH experiment that are
\emph{faithful} --- reproducing the observed singles rates, coincidence
rates, and unbiased marginals of a polarization singlet with symmetric
detector efficiency $\etaeff$ --- this work determines the minimal measurement
dependence $\md$ (in Hall's variational measure) required to achieve a
CHSH value $S$. Linear programming over the complete space of local
strategies yields, to machine precision at 32 grid points,
$\md(S,\etaeff)=\max\{0,\ \etaeff((S+2)\etaeff-4)/6\}$, whose edges
reproduce three landmark results: Hall's tight bound at $\etaeff=1$, the
Garg--Mermin detection threshold, and the postselection ceiling
$S=4/\etaeff-2$. At the quantum point the value is certified exactly:
$\md(2\sqrt2,\,9/10)=(27\sqrt2-33)/100$, with primal and dual
certificates verified in $\mathbb{Q}(\sqrt2)$ arithmetic. I further
solve, to $10^{-11}$, the unique detection profile
$D(m)=\sqrt{m}\,h(m)$ under which a deterministic sign model reproduces
the singlet correlation exactly, derive the $\sqrt m$ edge law
analytically, and prove that exact quantum correlations and
angle-independent coincidence rates are \emph{jointly impossible} for any
pure-detection model of this class. Mixtures of the solved profile with
measurement dependence are thereby placed on a quantitative trade-off between measurement
dependence and a specific observable: a $\cos 4(a-b)$ modulation of the
fourfold coincidence sum, of relative amplitude up to $12.4\%$, which
takes identical values at the two CHSH relative angles, harmonic by
harmonic; existing data exclude its full amplitude and bound its mixtures
with flat detection to a detection share of order $15\%$, and its
sub-percent content has never been bounded. I
give a settings-torus protocol whose log-Fourier analysis confines all
factorizable apparatus artifacts to the Fourier axes while the physical
signal occupies anti-diagonal pixels, with a two-setting cross-ratio and a full-torus variant, both
calibrated against two-channel efficiency imbalance, achieving $5\sigma$
sensitivity at $0.1\%$ within hours on standard entangled-photon
hardware. A flat result at $0.1\%$ excludes the solved detection family and its
mixtures to a detection share of order one percent; a fringe surviving the
calibrations of Sec.~\ref{sec:protocol} would be a joint $(a,b)$ dependence
that no single-pair, fixed-imperfection apparatus produces. Finally, I exhibit an explicit local mechanism ---
\emph{moir\'e phase locking} --- in which a hidden phase evolves
deterministically (no stochastic noise) through the standing beat
pattern of the measurement apparatus, and all randomness resides in a
quenched ensemble of frozen offsets and flight times. With a quenched ensemble of 1024 frozen offsets the mechanism
attains the certified floor \emph{exactly} at $\etaeff=1$ at every
tested fidelity, from uniform $10\%$ down to uniform $1\%$ of the
quantum cosine at the register angles (up to 161 pins), and, combined with the solved detection profile,
reaches zero measurement dependence at $\etaeff\approx0.80$,
consistent with the closed detection-only channel. The floor-attaining
census shares a low-order Fourier skeleton across independent
solutions and is not realizable by station-local offsets on the terrains
tested. Its residual correlation pattern, a softening
of the correlation extremes at the fidelity register, constitutes a
falsifiable fingerprint concentrated between the Bell angles.
\end{abstract}

\section{Introduction}
Bell's theorem \cite{Bell1964} and its CHSH form \cite{CHSH1969} exclude
local hidden-variable (LHV) accounts of quantum correlations only under
auxiliary assumptions, chief among them measurement independence (MI) and
fair sampling. Relaxations of these assumptions have been studied
quantitatively: Hall introduced measures of measurement dependence and
derived tight relaxed Bell inequalities \cite{Hall2010,Hall2011},
later extended to per-observer dependence \cite{FGHK2019}; detection
efficiency lies outside these relaxation axes --- a non-detection is not
an outcome in the $\pm1$ alphabet, so $\etaeff<1$ is not a species of
indeterminism $I$ --- and the detection loophole has its own threshold
literature
\cite{GargMermin1987,Eberhard1993,Pearle1970,GisinGisin1999,Larsson2014};
and exact measurement-dependence floors have recently been computed for
multipartite GHZ--Mermin correlations \cite{PriceOfLocality}. Trade-off
relations between measurement dependence and other resources have also
been derived \cite{Kimura2023,KimuraQuantum2025}. What has been missing
is the \emph{joint} price: the exact minimal measurement dependence as a
function of both the achieved violation and the detection efficiency, for
models constrained to reproduce what an experimenter actually observes.

This paper supplies that object (Sec.~\ref{sec:surface}), certifies it
exactly at the quantum point (Sec.~\ref{sec:cert}), and then follows the
mathematics to an experimental consequence. Solving the detection sector
exactly (Sec.~\ref{sec:mechanics}) leads to an impossibility result
(Sec.~\ref{sec:pinning}): no pure-detection model of this class ---
deterministic outcomes, a common detection profile, clicks independent given
the hidden phase --- can reproduce both the quantum correlation curve and
angle-independent coincidence rates. The
residue is a concrete observable --- a $\cos4(a-b)$ fringe in the total
(fourfold-sum) coincidence rate --- which takes identical values at the
CHSH angles (Sec.~\ref{sec:fringe}) and has been checked only at the
percent level (Sec.~\ref{sec:data}). I close with a validated
protocol (Sec.~\ref{sec:protocol}) whose two possible outcomes are both
informative: flatness bounds the detection share of the solved family; a fringe
surviving the calibrations of Sec.~\ref{sec:protocol} would be a joint
$(a,b)$ dependence that no single-pair, fixed-imperfection apparatus
produces.

Version~3 of this paper misdescribed the experimental record of the
fourfold sum, claimed that every local apparatus response factorizes, and
claimed that a flat sum forces measurement dependence. These are corrected
in Secs.~\ref{sec:pinning} and~\ref{sec:fringe}--\ref{sec:protocol}
following an independent review; the mathematical results of
Secs.~\ref{sec:surface}--\ref{sec:fringe} are unchanged, and
Theorem~\ref{thm:hide} is restated without change to its content.
Versions~4--7 took the imbalance calibration of Sec.~\ref{sec:protocol}
from the channel singles, with the coincidence marginals as its cross-check.
From version~8 the two roles are exchanged: background in the singles
biases the singles estimate toward zero, which leaves the imbalance term
under-corrected, and the coincidence tables alone determine the calibration
[Eq.~(\ref{eq:closed})]. The designs and their count requirements are
unchanged.

\section{Definitions and conventions}\label{sec:defs}
Settings $a\in\{a_1,a_2\}$, $b\in\{b_1,b_2\}$ at the standard angles
$(0^\circ,45^\circ;22.5^\circ,67.5^\circ)$. Outcomes $\pm1$ or no-click.
A deterministic local strategy assigns each party, for each of its two
settings, an element of $\{+1,-1,0\}$; there are $3^2\times3^2=81$ joint
strategies, and any LHV model (with or without measurement dependence)
reduces to four distributions $q_{ab}$ over them, one per setting pair.
Measurement dependence is quantified by Hall's variational measure
\cite{Hall2010}, here in the half-integral convention:
$\md=\max_{(ab),(a'b')}\tfrac12\sum_s|q_{ab}(s)-q_{a'b'}(s)|$; Hall's
$M$ [Eq.~(26) of Ref.~\cite{Hall2011}] omits the $\tfrac12$, so $M=2\md$
throughout, and his relaxed CHSH bound $B(0,0,M)=2+3M$ [Eq.~(47) of
Ref.~\cite{Hall2011}] reads $\md=(S-2)/6$ in my units; this
coarse-graining preserves the measure exactly. \emph{Faithfulness} means
the model reproduces the observed phenomenology of a maximally entangled
polarization singlet with symmetric detector efficiency $\etaeff$:
singles rates $\etaeff$ for each party, independent of both settings;
coincidence rate $\etaeff^2$, equal across setting pairs; vanishing
marginals both among detected singles and conditioned on coincidence; and
the stated coincidence correlators.

\section{The price surface}\label{sec:surface}
For each $(S,\etaeff)$, $\md$ is minimized over the four strategy
distributions subject to faithfulness and to the CHSH combination of the
coincidence correlators equalling $S$. This is a linear program in
$4\times81$ variables plus variational auxiliaries.

\begin{conjecture}[Price surface]\label{conj:surface}
For faithful local models at the standard CHSH settings,
\begin{equation}
\md(S,\etaeff)\;=\;\max\Big\{0,\ \frac{\etaeff\big((S+2)\etaeff-4\big)}{6}\Big\}.
\label{eq:surface}
\end{equation}
\end{conjecture}

\noindent\textbf{Status.} Verified against the exact LP to
$3.6\times10^{-16}$ at 32 grid points spanning $S\in[2.1,3.4]$,
$\etaeff\in[0.85,1]$; certified exactly at the quantum point
(Theorem~\ref{thm:cert}) and at six rational points
(Corollary~\ref{cor:ladder}); the general-$(S,\etaeff)$ statement awaits a
parametric dual certificate. Its three edges are known theorems: at
$\etaeff=1$, Eq.~\eqref{eq:surface} is Hall's tight line $\md=(S-2)/6$
[theorem of Ref.~\cite{Hall2011}, attained by the models of
Ref.~\cite{Hall2010}]; the zero set is the postselection ceiling
$S=4/\etaeff-2$; and at $S=2\sqrt2$ the section
$\md=\etaeff((1+\sqrt2)\etaeff-2)/3$ vanishes exactly at the Garg--Mermin
threshold $\etaeff=2(\sqrt2-1)$ \cite{GargMermin1987}. The equivalent
inverted form
\begin{equation}
S\;\le\;\Big(\frac{4}{\etaeff}-2\Big)+\frac{6\md}{\etaeff^{2}}
\end{equation}
exhibits the interaction of the two relaxations: the measurement-dependence
budget is amplified by $1/\etaeff^2$, i.e., it acts on the coincidence
subensemble. The naive additive combination
$\md=(1{+}\sqrt2)/3-2/(3\etaeff)$, which shares both endpoints, was
registered as a hypothesis in advance and is \emph{falsified}: the true
surface lies strictly below it in the interior. A corollary worth highlighting:
$\md=0$ whenever $S\le 4/\etaeff-2$, so measurement dependence is implied
only by correlations in the violating region --- ordinary (classical)
statistics implies none.

\begin{figure}[t]\centering
\includegraphics[width=.72\linewidth]{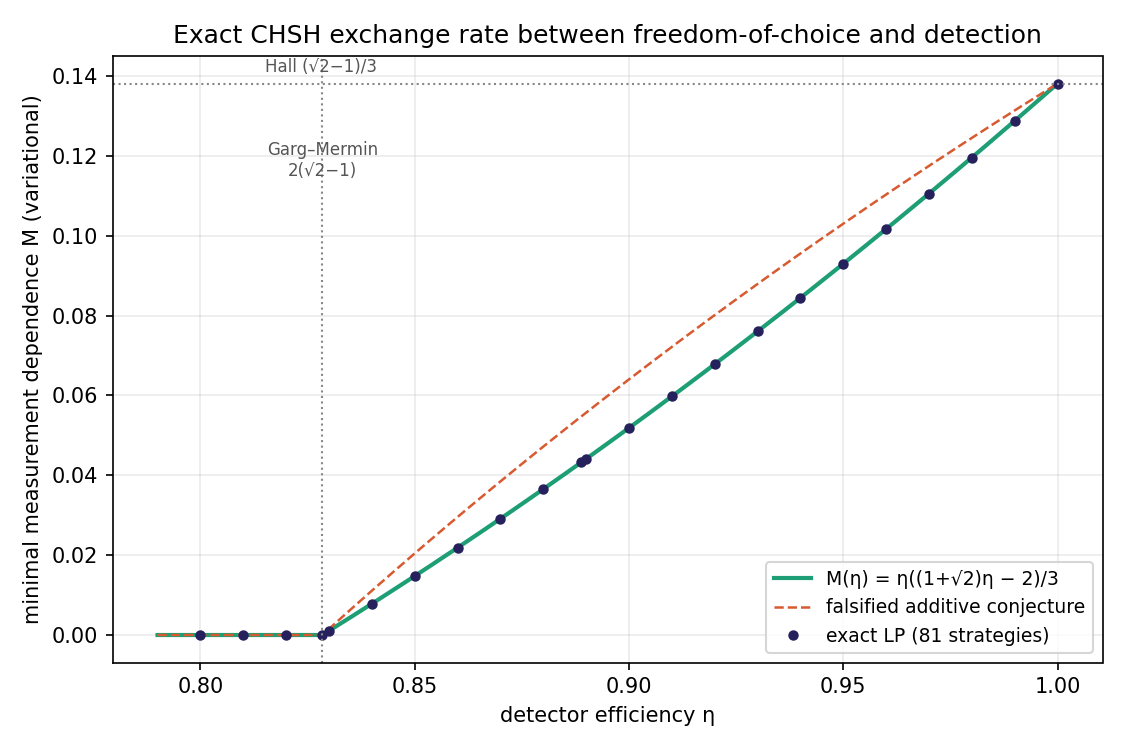}
\caption{The $S=2\sqrt2$ section of the price surface: exact LP values
(points), the closed form of Eq.~\eqref{eq:surface} (solid), and the
falsified additive hypothesis (dashed). Landmarks: Hall's floor and the
Garg--Mermin threshold.}\label{fig:surface}\end{figure}
\section{Exact certification at the quantum point}\label{sec:cert}
\begin{theorem}\label{thm:cert}
At $S=2\sqrt2$, $\etaeff=9/10$, the minimal measurement dependence of
faithful local models is exactly
$\md=(27\sqrt2-33)/100$.
\end{theorem}
\begin{proof}[Proof sketch]
A floating-point solve of the LP identifies the optimal support; every
nonzero primal weight and dual multiplier is then reconstructed in
$\mathbb{Q}(\sqrt2)$ by integer-relation detection and verified
symbolically: the primal candidate (68 weights) satisfies all 36
faithfulness equalities exactly, is nonnegative, and attains the value;
the dual candidate is exactly feasible with matching objective. Upper and
lower certificates together fix the value. All verification steps are in
exact arithmetic; floating point is used only to propose candidates.
Scripts and certificates accompany the paper.
\end{proof}
\noindent At the optimum, all six pairwise variational distances between
the four conditional distributions are \emph{equal} --- the extremal
model balances its measurement dependence democratically across setting
contexts --- a symmetry expected to organize the parametric dual family
for Conjecture~\ref{conj:surface}.

\begin{corollary}[Rational ladder]\label{cor:ladder}
By the same certification pipeline, in plain rational arithmetic:
$\md(21/10,1)=1/60$, $\md(5/2,1)=1/12$, $\md(5/2,19/20)=209/4800$,
$\md(27/10,9/10)=69/2000$, $\md(29/10,19/20)=2489/24000$, and
$\md(17/5,17/20)=1003/12000$ --- each an exact theorem with verified
primal and dual certificates, and each equal to Eq.~\eqref{eq:surface}.
\end{corollary}

\section{The exact detection sector}\label{sec:mechanics}
Consider deterministic sign outcomes with probabilistic detection: shared
phase $\lambda\sim U[0,\pi)$, outcome $\mathrm{sgn}\cos2(\lambda-a)$
(partner carries $\lambda+\pi/2$), and click probability $D(m)$ with
$m=|\cos2(\lambda-a)|$, clicks independent given $\lambda$. Writing
$x=2\lambda$, $g(x)=D(|\cos x|)$, and $f(x)=\mathrm{sgn}(\cos x)\,g(x)$,
the coincidence correlator is a ratio of circular autocorrelations:
$E(\Delta)=-C_f(2\Delta)/C_g(2\Delta)$.
Exact reproduction of the singlet is therefore the functional equation
\begin{equation}
C_f(\theta)\;=\;\cos\theta\; C_g(\theta)\qquad\forall\theta.
\label{eq:master}
\end{equation}

\begin{lemma}[Edge law]\label{lem:sqrt}
Any solution of \eqref{eq:master} with $D$ continuous has
$D(m)\sim c\,\sqrt m$ as $m\to0$.
\end{lemma}
\begin{proof}[Proof sketch]
Writing $\varphi(t)=D(|\sin t|)$, Eq.~\eqref{eq:master} is equivalent to
$4(\varphi*\varphi)(\theta)=(1-\cos\theta)A(\theta)$, where the left side
integrates over the sign-disagreement arcs of width $\theta$ and
$A$ is the global autocorrelation. If $D(m)\sim cm^{p}$, the left side
scales as $\theta^{2p+1}$ while the right side scales as $\theta^{2}$;
matching forces $p=1/2$.
\end{proof}

Solving \eqref{eq:master} with the edge law built in
($D=\sqrt m\,h(m)$, $h$ a Chebyshev series, Gauss--Jacobi$(\tfrac12,\tfrac12)$
quadrature on the arc form) yields a solution with maximal residual
$1.5\times10^{-12}$; the corresponding correlation curve matches
$-\cos2\Delta$ to $4\times10^{-11}$, verified independently by Monte
Carlo. The factor $h$ is smooth and monotone with
$h(0)/h(1)=1.17450854538894$; no closed form has been identified
(the candidate $3\pi/8$ is excluded at $10^{-12}$). The maximal-scale
efficiency of the exact profile is $80.1\%$, below the Garg--Mermin
ceiling as required.

\begin{figure}[t]\centering
\includegraphics[width=.95\linewidth]{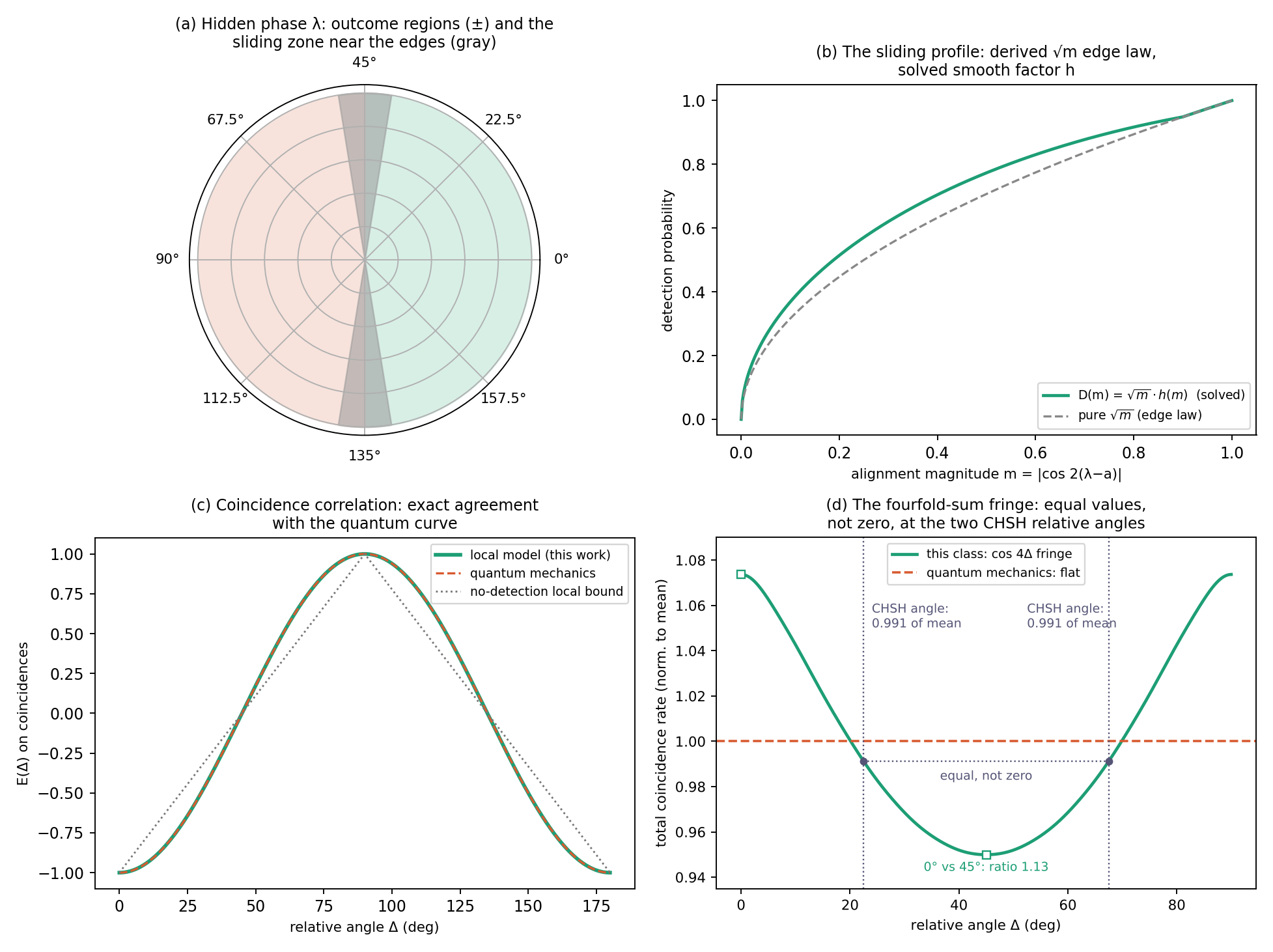}
\caption{The detection sector. (a) Hidden-phase outcome regions and the
sliding zone; (b) the solved profile $D(m)=\sqrt m\,h(m)$ with its
derived edge law; (c) exact agreement of the coincidence correlation with
the quantum curve; (d) the fourfold-sum fringe, normalized to its mean, taking equal values
($0.991$ of the mean) at the two CHSH relative angles (Theorem~\ref{thm:hide})
and its full peak-to-peak between $0^\circ$ and $45^\circ$.}\label{fig:mech}\end{figure}
\section{Squares, spheres, and an impossibility result}\label{sec:pinning}
The harmonic weights of $C_f$ and $C_g$ are squared Fourier amplitudes,
hence nonnegative: whenever the coincidence rate is angle-independent,
$E(\Delta)$ is (the negative of) a positive-definite function, which by
Schoenberg's theorem \cite{Schoenberg1942} is a Gram function of unit
vectors, and Tsirelson's bound \cite{Tsirelson1980} follows for
\emph{all} angle choices. This was verified adversarially: over 300 random
positive harmonic mixtures with angles optimized against the bound, CHSH
never exceeded $2\sqrt2$. Conversely, super-quantum postselected values
in this model class are purchased with angle-\emph{dependent} coincidence
rates (a hard edge-band model reaching $S=4$ shows $48\%$ rate
modulation).

\begin{proposition}[Impossibility in the independent-click class]\label{prop:imposs}
No detection profile $D$ yields both the exact singlet correlation
$E(\Delta)=-\cos2\Delta$ and an angle-independent coincidence rate.
\end{proposition}
\begin{proof}
The coincidence rate is proportional to $C_g(2\Delta)$; it is constant
iff all nonzero harmonics of $g$ vanish, i.e., $D$ is constant; but
constant $D$ gives the sawtooth correlation, not the cosine.
\end{proof}

Pearle noted the same setting dependence of the pair-detection rate in his
model \cite[Sec.~V, p.~1421]{Pearle1970}; Gill \cite{Gill2020} calls it the
model's ``major defect'' and conjectures, with numerical evidence, that no
distribution in Pearle's class removes it. Proposition~\ref{prop:imposs}
proves the counterpart for independent clicks on the circle. Outside that
class, faithful local models with flat rates do exist for
$\eta_{\rm eff}\le 4/(2+\pi)\approx0.778$ (\cite{Larsson1999}, and a
symmetrized form of \cite{GisinGisin1999}; Sec.~\ref{sec:data}); the
fringe is a property of the class, not of local detection models in
general.

\section{A mechanism that saturates the surface}\label{sec:mpl}
The results above bound and dissect the space of faithful local
accounts; this section exhibits an explicit member of that space which
saturates the bound. The construction, which I call \emph{moir\'e
phase locking}, involves no stochastic dynamics and no fitting of the
physics: a single hidden phase $\lambda\in[0,\pi)$ evolves
deterministically under
\begin{equation}
\dot\lambda \;=\; v \;-\; U'(\lambda;\,a{+}\xi_A,\,b{+}\xi_B),
\qquad
U=-\!\!\sum_{(m,n,\pm)}\! w_{mn\pm}\,
\cos\!\big(2k\lambda-\phi_{mn\pm}(a,b)\big),
\label{eq:mpl}
\end{equation}
where $k=m\pm n$ and the phases $\phi_{mn\pm}=2m(a{+}\xi_A)\pm
2n(b{+}\xi_B)\,(+n\pi)$ are fixed by the geometry of the beat
(moir\'e) pattern between the pair's internal structure and the two
analyzers; $v$ is the residual detuning between the two patterns'
rates. All randomness is \emph{quenched}: each event carries frozen per-station offsets $(\xi_A,\xi_B)$ and a frozen evolution
time $\tau$ (physically, a path-length spread), drawn once from a
setting-independent joint density $\rho$. The event's phase
distribution is the deterministic pushforward of uniform initial
conditions --- a residence-time law $\propto 1/(v-U')$ in the running
regime, a point mass at the stable lock when captured.

With the interaction weights selected by greedy search under the
faithfulness objective (23 active beat terms, dominated by asymmetric
cross-order harmonics) and $\rho$ optimized by linear programming over
a quenched offset ensemble (a $16\times16$ grid, or 1024 random or
stratified offsets), five flight times, and three coupling
scales, the mechanism's minimal measurement dependence at
$S=2\sqrt2$, $\etaeff=1$, with the correlation curve at the
constrained analyzer orientation held uniformly (41 angles on
$[2^\circ,178^\circ]$ at $4.4^\circ$ spacing) within tolerance
$\varepsilon$ of $-\cos2\Delta$, depends on the coverage of the
ensemble. On a $16\times16$ lattice of offsets (256 atoms) the prices at uniform
tolerance $\varepsilon=0.10,\,0.05,\,0.02,\,0.01$ are
$1.000,\,1.015,\,1.235,\,1.375$; a half-cell translation of the
lattice reproduces these to within $0.03$, and random or stratified
ensembles of the same size pay $12$--$14\%$ more at
$\varepsilon=0.10$ and $0.15$--$0.20$ more at $\varepsilon=0.01$
(three seeds, spread $\le0.03$). At 1024 offsets, random or
stratified, the mechanism attains the certified floor exactly ---
$\md/\md_{\min}=1.000$, with $\md_{\min}=(\sqrt2-1)/3$ Hall's
floor --- at every tolerance down to
$\varepsilon=0.01$ (two samplers, two seeds each). The fidelity
premium of the 256-atom ensemble is therefore a resolution artifact; at
sufficient coverage the mechanism saturates Hall's bound at $1\%$
register-uniform fidelity.

The 256-atom prices above are exact, deterministic properties of the stated
discrete-time dynamics (explicit Euler, $\mathrm{d}t=0.004$), and each
is cross-validated to the displayed precision by an independently coded
linear program over the same register ($1.375$ at $\varepsilon=0.01$
by both solvers; $1.235$ at $\varepsilon=0.02$ independently recovered
by the detection-composition program at $\alpha=0$). They are not
quoted as continuum limits: at 256 atoms, halving the timestep shifts the
$\varepsilon=0.01$ price from $1.375$ to $1.345$; at 1024 atoms the
price is pinned at the bound, which is the theorem's floor and cannot
move below it.
The residual correlation pattern of the $\varepsilon=0.01$ solution is
shown in Fig.~\ref{fig:pattern}; its robust content is analyzed in the
caption. The same sign envelope --- positive residuals where
$E_{\rm QM}\to-1$, a negative plateau on $76^\circ$--$116^\circ$ ---
is returned by the lattice-free 1024-offset ensemble at the floor (mean
residual $+0.0096$ and $-0.0089$ respectively, band-riding at $90\%$
of register angles), so the pattern is a property of the mechanism at
Hall's bound rather than of the sampling. The mechanism is exactly rotationally covariant: under a global
rotation $a\to a+s$, $b\to b+s$, $\lambda\to\lambda+s$, every phase
$\phi_{mn\pm}$ in Eq.~\eqref{eq:mpl} shifts by $2(m\pm n)s$, matching
the shift $2ks$ of the beat argument in both sign branches
($k=m\pm n$), so $E(a+s,b+s)=E(a,b)$ identically; numerically the
solved solutions reproduce the same correlation table to four decimals
on the slices $a=0^\circ,22.5^\circ,45^\circ,67.5^\circ$. (Version~2
of this paper misattributed a $0.065$ residual to broken covariance;
that residual appears identically on the constrained slice and is a
register effect, as follows.) (Throughout, the CHSH combination itself is enforced as an exact
equality constraint on the four standard-setting correlators,
independent of the register band, so $S=2\sqrt2$ does not rest on the
melody register.) The melody constraints pin the curve at
the register angles only, and between pins the deterministic locking
dynamics produce excursions. At the $16\times16$ ensemble these reach
$0.065$ between the $4.4^\circ$-spaced pins, and densifying the
register raises the price --- $1.375,\,1.511,\,1.635\times$ the floor
at $41,\,81,\,161$ pins, $\varepsilon=0.01$ --- while barely reducing
them: the per-column correlation curves share a locking staircase that
$256$-atom mixtures cannot smooth. At 1024 random offsets both effects
resolve together: the mechanism attains the certified floor at the
161-pin register ($\md=0.1380711874577208$ at $\varepsilon=0.01$ against
$(\sqrt2-1)/3=0.1380711874576983\ldots$, agreement to twelve
significant digits; the support
spreads to 656 active atoms with the twelve largest carrying only
$10\%$ of the weight), so the density premium, like the fidelity
premium, is a coverage artifact; and the same solution's between-pin
residuals, sampled at $0.2^\circ$ resolution across the previously
worst window and at fifteen points spanning the full range, fall to
$\le0.026$ and $\le0.017$ respectively, following the same
extreme-softening envelope as Fig.~\ref{fig:pattern}. Uniform
continuum fidelity is not claimed: fidelity at the pins is
$\varepsilon$ by construction, off-pin fidelity is sampled rather
than bounded, improves jointly with offset coverage and register
density, and its continuum limit is open. Two ingredients of the mechanism are found rather than derived, and they
live in different places. The terrain $U$ of Eq.~\eqref{eq:mpl} is fixed by
the analyzers' beat phases and was selected by greedy search under the
faithfulness objective; the census $\rho$ over frozen offsets, flight
times and coupling scales is fixed at the pair's creation and was chosen by the linear program.
Deriving either from a formation process is an open problem, which I state
as such; what the computations below establish is what the census
looks like at the floor on this terrain, which constrains any such
derivation.

\paragraph{The structure of the optimal census.}
The floor-attaining census is not an arbitrary nonuniform measure. Four
independent 1024-offset solutions at Hall's floor --- random seeds 0 and 1
and stratified seeds 0 and 1, each at the 161-angle register and
$\varepsilon=0.01$, each with $\md/\md_{\min}=1.000$ to twelve digits, and
sharing no offsets --- have the same low-order Fourier content in the
station offsets $(\xi_A,\xi_B)$. The dominant term is
$\cos\!\big(4(\xi_A+\xi_B)-\varphi\big)$ (the census phase $\varphi$ is
distinct from the terrain phases $\phi_{mn\pm}$ of Eq.~\eqref{eq:mpl}),
with amplitude $|F|=0.35$--$0.38$ against a
permutation-null 95th percentile of $0.17$--$0.18$ (weights fixed, offsets
re-paired, 400 draws), and with $\varphi$ agreeing across the four
solutions to within $5^\circ$ ($156.5^\circ$, $161.1^\circ$,
$157.1^\circ$, $158.0^\circ$; the chance of four uniform phases falling in
that arc is below $10^{-5}$). The $\xi_A$ marginal is nonuniform, with
a $\cos4\xi_A$ component of amplitude $0.24$--$0.31$, and a weaker relative-phase term
$\cos\!\big(4(\xi_A-\xi_B)-\psi\big)$ is present at amplitude
$0.12$--$0.15$, about a third of the sum term, above its own permutation
null on every solution ($p = 0.005$--$0.04$) and with $\psi$ agreeing on
three of the four; the anti-diagonal share of the $|k|\le6$ plane as a
whole sits at the null. The census is therefore not a function of
$\xi_A-\xi_B$ alone, but it is not independent of it either. The A/B asymmetry is
not an artifact of the solve frame: re-solving with $b$ pinned and $a$
swept over the register, on the same offsets, gives the same floor and
the same census skeleton (leading phase within $1^\circ$), with the
asymmetry on the same side; whether it is a property of the
terrain as bred, which was itself selected in the pinned frame, is open. At the floor, then, the low-order skeleton of the census is
a function of the terrain in the pinned frame, realized by different
offsets and different higher-order content in different runs; the
optimizer's freedom is the freedom to realize that skeleton. Perturbing the terrain
illustrates the division of labour: shifting the phase of the $(2,2,+)$ beat
term by $+45^\circ$ moves the $(2,2)$ pixel of the census by $-38^\circ$
(the sign the convention $\sin(2k\lambda-\phi)$ dictates) and leaves
$\md-\md_{\min}=5\times10^{-8}$; removing the term leaves
$\md-\md_{\min}=3\times10^{-14}$ and the linear program finds a different
census with a $(2,2)$ pixel at a different phase. The census skeleton is a
function of the terrain; the price is not.

The skeleton also bears on where $\rho$ cannot come from. On the random0
offset banks, no census of the station-local form
$q(\tau,s)\,f(\xi_A)\,g(\xi_B)$ --- nor the weaker class in which $\xi_B$
is independent of everything else --- reproduces the singlet correlators
at $\varepsilon=0.01$ ($N_B=32$ bins, 14 starts per class; the
correlators missing their targets by $3$--$13\%$ on every start), while a binned joint
census does ($1.12$ times the floor, the excess being the binning); where
the station-local class is feasible, at $\varepsilon=0.05$, it pays
$2.2$--$3.1$ times the floor. A terrain bred to favour
station-local censuses (five greedy rounds at breeding scale, scored by
the product-class price) does not change this at full scale (product
class $2.86$ times the floor at $\varepsilon=0.05$, infeasible at
$0.01$). On the paper's terrain --- on four censuses by their Fourier
content and on the random0 banks by direct test --- and on one terrain
bred against the conclusion, the frozen offsets of the two stations must
therefore be correlated at the source: no per-station origin of $\rho$
--- a property of the analyzer, of the collection geometry, or of
detector-side processing --- suffices on the terrains tested, and the reason none should is the one Bell's
argument rests on: station-side processing is part of $A(a,\lambda)$. The formation process to be derived acts at pair creation
and must produce the skeleton above relative to the apparatus's beat
structure; the four floor-attaining censuses carry the relative phase
$\xi_A-\xi_B$ only as the weaker term above, subordinate to the sum. Nor
does the census factorize in the sum and difference of the offsets: on
the same banks, a census of the form
$q(\tau,s)\,f(\xi_A+\xi_B)\,g(\xi_A-\xi_B)$ also fails to reproduce the
correlators at $\varepsilon=0.01$ from every one of 14 starts (missing
its targets by $2$--$4\%$, against $3$--$13\%$ for the station-local
form) and pays $2.75$ times the floor at $\varepsilon=0.05$ (against
$3.07$); a census depending on $\xi_A-\xi_B$ alone is a member of that
class. The sum carries the dominant structure, but the sum and the
difference of the offsets are not stamped independently. Scripts, offset
banks, the four censuses, the perturbation runs and the station-local
and sum/difference tests are deposited at \url{https://doi.org/10.5281/zenodo.23537610} and
in the repository (Code and data availability).

The floor-attaining census is also not smooth. On a $32\times32$ lattice
of offsets (exact quadrature below the sixteenth harmonic), no census
band-limited to $|k|\le5$ in $(\xi_A,\xi_B)$ reproduces the singlet to
$1\%$: the best misses by $1.06\%$ ($1.17\%$ at $|k|\le4$, $8.5\%$ at
$|k|\le3$; certified minima of the linear program, which hold on 1024- and
4096-offset random samplings as well), and the smoothest census that does
attain it pays $1.70$ times the floor, falling to $1.43$, $1.26$, $1.11$
and $1.00$ at $K=8$, $10$, $12$, $14$ (the lattice resolves no harmonic
above the sixteenth, so its curve must close by $K=16$; on 4096 random
offsets, which have no such ceiling, the price is $1.01$ at
$K=18$).\footnote{Prices of band-limited censuses on
random offsets include a quadrature cost: a census band-limited to $K$
harmonics evaluated on $N$ random atoms computes its correlators as an
$N$-point random quadrature, and the linear program spends measurement
dependence cancelling that noise ($2.27$ at $K=6$ on 1024 random offsets,
$1.83$ on 4096); the lattice values are free of this cost. Version 3 of
the deposit reports the 1024-random values, version 4 the lattice.} The
floor census carries angular structure of $5$--$6^\circ$ in the offsets.
The $1\%$ tolerance is a convenience, not a requirement: the price linear
program with a zero-width correlator band is feasible at the floor on the
lattice census.

The census also depends on the $(\tau,s)$ cell (flight time and
coupling scale). No census that stamps
the same offsets at every cell and only reweights them was found at the
floor: the best found pay $1.47$ times the floor with one shared offset
density and $1.20$ with two (alternating linear programs from six starts;
local optima, so upper bounds on the class's price), and the cell
densities of the smoothest attaining census ($K=6$) are more alike within
a coupling scale than within a flight time. The lattice and 4096-offset
censuses, the band-limited and cell-blind programs and their transcripts
are in version 4 of the same deposit.

The detection rule has so far been held fixed at
$\mathrm{sgn}\cos2(\lambda-a)$. I varied it: with the $\lambda$-histograms
of every (flight time, coupling scale, setting, offset) stored, any
stochastic response $D(\lambda-a)$ with $|D|\le1$ is a re-contraction of
the same trajectories, and a family of profiles --- smoothed edges, domed,
tilted and notched caps, asymmetric and kinked walls, staircase walls ---
was priced on the lattice census (deposit version 6). A smooth edge of
width $10.6^\circ$ (the band in which $|D|<0.9$) costs $1.048$ times the
floor, $14.2^\circ$ costs $1.457$, and $21.5^\circ$ loses the singlet;
any feature inside the pass band --- a domed or tilted cap, a notched one
--- loses the singlet at every census bandwidth, with a miss that grows
linearly with the feature's size for the dome and the tilt; a three-level
reader that is undecided within $\pm3^\circ$ of the boundary attains the
floor exactly, while $\pm6^\circ$ costs $1.94$; and no shaped reader
lowers the census's bandwidth demand --- every one raises the $K=6$ price
above $1.70$, to $2.2$--$2.6$ where the singlet remains reachable at
$K=6$ and beyond reach otherwise. The sign rule is therefore not an idealization the census
tolerates but a requirement, up to an undecided band of about $6^\circ$
at the boundary.

The census has also been assumed independent of the pair's initial phase
$\lambda_0$, taken uniform. If the birth is instead a joint density over
$(\lambda_0,\xi_A,\xi_B,\tau,s)$, with $\lambda_0$ carrying one or two
harmonics whose coefficients are band-limited in the offsets at the
census's own degree, the smoothness premium falls: at $K=6$ from $1.700$
to $1.373$ with one harmonic and $1.230$ with two, and at $K=8$ from
$1.429$ to $1.164$ (linear programs with positivity of the joint density
on a $\lambda_0$ grid and measurement dependence measured, as throughout,
on the distribution of the final phases; each price bracketed to
$10^{-4}$ by a clipped, re-measured census that meets the band to within
$2\times10^{-5}$ of $\varepsilon$; deposit version 6). A joint
birth that depends on the cell alone, or that is smoother than the census
in the offsets (degree 3), buys nothing or almost nothing ($1.700$;
$1.647$). The fitted birth is a half-wave in $\lambda_0$ whose preferred
phase is locked to the approach phase of the pinned station,
$\lambda_0\approx77^\circ-\xi_A$, and by the model's
$+90^\circ$ symmetry the census carries the even harmonics of the offsets
while this lean carries the odd ones. This enlarges the class of
Sec.~\ref{sec:mpl}: the floor does not move, but the structure the birth
must carry is shared between the offsets and the initial phase. The
lean's fine structure is, so far, a lattice result.

In the terminology of Hance, Hossenfelder and Palmer \cite{HHP2022}, a
model whose measurement dependence resides in a non-trivial measure on the
hidden-variable space, rather than in a causal link between settings and
hidden variables, is \emph{supermeasured}. The mechanism here is closest to
that picture in spirit: its frozen variables are drawn from a
setting-independent density $\rho$ --- one tied, by the results above,
to the apparatus's beat structure --- and Eq.~\eqref{eq:surface} is the
least measurement dependence the resulting phase distribution must carry. It differs in that the settings enter
the evolution of $\lambda$ through Eq.~\eqref{eq:mpl}, so whether it is
supermeasured in their strict sense, or a causal form of measurement
dependence, is left open here. The problem of deriving $\rho$ is the one
that Garza and Hance \cite{GarzaHance2025} address for a different model,
obtaining the measure from the limited precision of measurements of
quantities that do not commute with conserved quantities; their model, like
this one, departs from quantum mechanics by a testable residual.

\begin{figure}[t]
\centering
\includegraphics[width=0.98\linewidth]{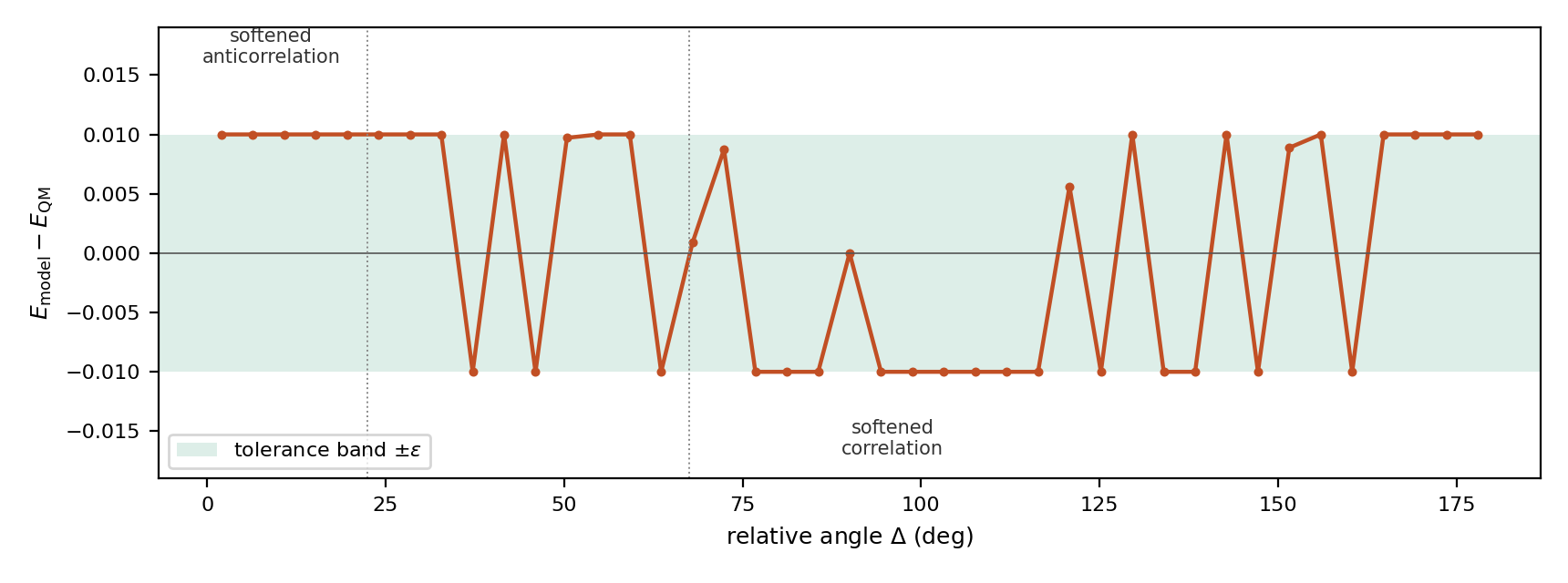}
\caption{Residual correlation pattern
$E_{\rm model}-E_{\rm QM}$ of the mechanism at uniform tolerance
$\varepsilon=0.01$ ($16\times16$ ensemble; 41 pinned angles; optimal solutions ride the
tolerance band, as linear-programming vertices must). The robust
content is the sign envelope: positive residuals where
$E_{\rm QM}\to-1$ ($\Delta\lesssim33^\circ$,
$\gtrsim164^\circ$) and a negative plateau where
$E_{\rm QM}\to+1$ ($76^\circ\!-\!116^\circ$): the mechanism
\emph{softens the correlation extremes}, paying its tolerance where
perfection is demanded. This visibility-like signature is the mechanism's falsifiable
fingerprint: $\le1\%$ at the register angles, with the same envelope
continuing between the pins, where the floor-attaining 1024-offset,
161-angle solution shows sampled residuals up to $0.017$ (full range)
and $0.026$ (worst window); see Sec.~\ref{sec:mpl}.}
\label{fig:pattern}
\end{figure}

Finally, the mechanism composes with the detection sector of
Sec.~\ref{sec:mechanics}. Reweighting coincidences by
$D_\alpha=(1{-}\alpha)+\alpha D^{*}$ and re-optimizing $\rho$ under
rate-normalized constraints with flat coincidence rates enforced, the
combined system interpolates between the two closed channels: at
$\alpha=0$ (perfect detectors) the $16\times16$ source ensemble pays
$1.235\times$ the floor at $\varepsilon=0.02$ ($1.000$ at 1024
offsets, as above); at $\alpha=1$ the linear program returns
$\md=0$ at $\etaeff=0.797$, consistent with the detection-only
channel of Sec.~\ref{sec:mechanics} (certified ceiling
$\etaeff\approx0.801$; the residual half-percent gap is the
interpolation error of the sampled $D^{*}$ profile embedded in the
scan, whose mean efficiency is $0.7967$). Local accounts of the singlet thus occupy a
one-parameter dial between an informed source and an inefficient
detector, with every intermediate position priced by
Eq.~\eqref{eq:surface}.

\section{The fringe and the hiding theorem}\label{sec:fringe}
For the exact profile, the total coincidence rate obeys $R_d(\Delta)\propto
C_g(2\Delta)$, a $\cos 4\Delta$-dominated fringe (leading harmonic to first
overtone in the ratio $5.87{:}1$) of relative peak-to-peak amplitude
$12.4\%$; its harmonic content is
$R_d/\bar R_d = 1 + 0.0577\cos 4\Delta + 0.0098\cos 8\Delta + 0.0031\cos 12\Delta + \cdots$.
A hybrid of this profile with flat detection paying the
remainder of its bill through measurement dependence shows a proportionally
reduced amplitude, quantified by the LP along the interpolating family of
Sec.~\ref{sec:mpl} (the ``guilt trade-off'' of Fig.~\ref{fig:guilt}): within that family, a bound $\delta_f<2\%$ on the fringe forces
$\mathcal{M}\ge 45\%$ of Hall's floor; $\delta_f<1\%$ forces $68\%$;
$\delta_f<0.5\%$ forces $80\%$. These are statements about the family:
faithful pure-detection models with flat rates exist outside the class of
Proposition~\ref{prop:imposs} for $\eta_{\rm eff}\le0.778$ (Sec.~\ref{sec:data}), so a flat fringe bounds
the detection share of this family and does not by itself impose measurement
dependence.

\begin{theorem}[Hiding theorem]\label{thm:hide}
The fringe takes identical values at the two CHSH relative angles: for every harmonic
$k$, $\cos(2k\cdot45^\circ)=\cos(2k\cdot135^\circ)$, hence
$R_d(22.5^\circ)=R_d(67.5^\circ)$ exactly, term by term.
\end{theorem}
\noindent Four-point experiments at the CHSH angles are therefore blind
to the observable; designs at relative angles $0^\circ$ and $45^\circ$ are
not (Sec.~\ref{sec:protocol}).

\begin{figure}[t]\centering
\includegraphics[width=.68\linewidth]{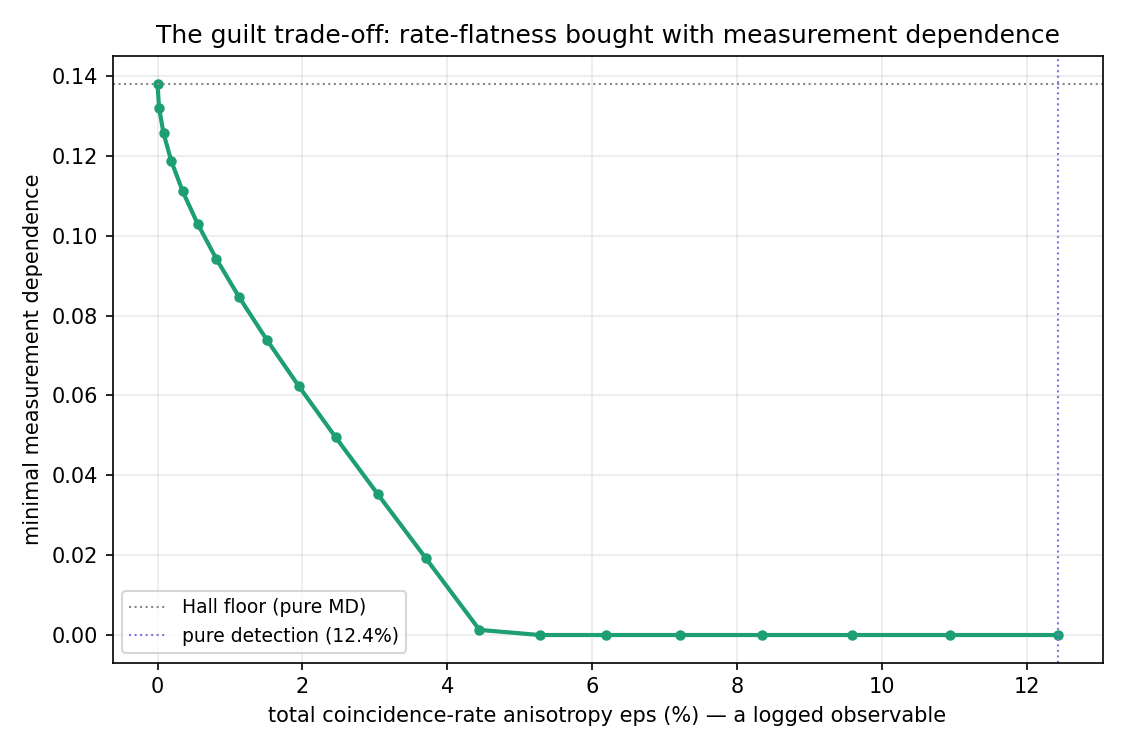}
\caption{The trade-off between the two channel signatures: minimal
measurement dependence versus fourfold-sum anisotropy along the
interpolating detection family.}\label{fig:guilt}\end{figure}
\section{The observational record}\label{sec:data}

The fourfold sum $R_{++}+R_{+-}+R_{-+}+R_{--}$ has long been recognized as
the natural fair-sampling diagnostic
\cite{AdenierKhrennikov2003,Adenier2008,AdenierAJP2008}, and its constancy has
been checked, though not at the sub-percent level. Aspect, Grangier and
Roger state that the sum of the four coincidence rates was observed to be
constant when the orientations were changed, and that the sum of one photomultiplier's coincidences with both channels on the other side was constant \cite[p.~92]{Aspect1982}; Aspect's thesis
records that the invariance of the fourfold sum was checked in every series of measurements,
within statistical precision: in the 26-run series from which its
Table~VIII-2 is taken (that table is at parallel settings) the total
true-coincidence rate ranged over $78.8$--$81.2\,{\rm s}^{-1}$ against
a per-run standard deviation of $1.3\,{\rm s}^{-1}$ \cite[Sec.~VIII-5c]{AspectThesis1983}, and
three further series totalling $5000\,{\rm s}$ were taken with the polarizers
in varied orientations \cite[Sec.~VIII-5b]{AspectThesis1983}. The angles within each series are
not tabulated, so the check is a routine invariance test at roughly
$\pm1.5\%$ in the one series quantified, rather than a bound at each relative
angle. It is nevertheless hard to reconcile with the full-amplitude profile
of Sec.~\ref{sec:mechanics}, which predicts $R_d(0^\circ)/R_d(45^\circ)=1.13$ ---
about $86$ against $76\,{\rm s}^{-1}$ at Aspect's rate. A setting-independent
local loss rescales $D\to tD$ and leaves the relative fringe unchanged, so
the low collection efficiency of that experiment does not hide it. Aspect's
later review repeats the observation and adds that constancy of the sum is
consistent with fair sampling but not a proof of it \cite[Sec.~9.4]{AspectUnspeakables2002}. The
high-efficiency photonic tests
\cite{Christensen2013,Giustina2013,Giustina2015,Shalm2015} closed the fair-sampling loophole
in the Eberhard/CH form \cite{Eberhard1993}, with one detector per side, and so did not record
$R_d$.

The single-channel scan of Brida, Genovese and Piacentini \cite{Brida2007} constrains the
same profile more tightly. Within the model class of Sec.~\ref{sec:mechanics} the exact singlet
correlation and unbiased marginals among coincidences give
$R_{++}=R_{--}=R_d(1+E)/4$ and $R_{+-}=R_{-+}=R_d(1-E)/4$, so a
$\cos 4\Delta$ term in $R_d$ multiplies the $\cos^2$ fringe of a
single-channel measurement and appears there at first order. Brida et al.\
scanned the coincidence rate at eight relative angles in steps of $\pi/8$
with visibility $0.978$ and report, for Santos's deviation statistic \cite{Santos2004},
$D(\eta)-\Delta_{\min}=0.0073\pm0.0022$; Santos's analysis of the same data
gives the experimental rms deviation of the normalized fringe from
$1+V\cos2\varphi$ as $\Delta_{\rm exp}=0.0074$ \cite{Santos2007}. The Sec.~\ref{sec:mechanics} profile at
those eight angles gives $0.053$. Their $V_B/V_A$ statistic, by
contrast, uses only the angles $0^\circ$, $22.5^\circ$, $67.5^\circ$ and
$90^\circ$ and is blind to the fringe, as Theorem~\ref{thm:hide} predicts. Taken together,
existing data exclude the full-amplitude profile and bound its mixtures with
flat detection to a detection share of order $15\%$ (about $20\%$ at two standard
deviations); the sub-percent region below that is unmeasured.

Flatness of the sum does not, by itself, force measurement dependence.
Larsson \cite{Larsson1999} gives a local model with an angular hidden variable in which
the single-detection probability is $\eta$ and the pair-detection
probability is $\eta^2$ at every pair of settings, reproducing the singlet
statistics for $\eta\le 4/(2+\pi)\approx0.778$ \cite[Eqs.~(21)--(22)]{Larsson1999}, and
a coin-symmetrized form of the planar Gisin--Gisin model \cite{GisinGisin1999}, with a common
pair-loss factor, does the same at the same efficiency. Both satisfy the
faithfulness conditions of Sec.~\ref{sec:defs} with flat rates and $\mathcal{M}=0$, and
both lie outside the class of Proposition~\ref{prop:imposs}: Larsson correlates the two
clicks through a second shared variable, and the symmetrized Gisin--Gisin
model uses different click functions on the two sides in each branch. Adenier and
Khrennikov made the same point in proposing a fair-sampling test \cite{AdenierKhrennikov2003}:
invariance of the sum ``cannot be taken as a proof that there is no unfair
sampling,'' and Larsson's review states that fair sampling is not equivalent
to constant detector efficiency, citing the same models
\cite[Sec.~3.1]{Larsson2014}. The anomalies reported in the raw Innsbruck data \cite{Weihs1998,AdenierKhrennikov2007}
concern the marginals --- an apparent dependence of one party's marginal on
the other's setting under fair sampling --- not a setting dependence of the
fourfold sum, and their authors state that they cannot be produced by
variations of the single-count channel efficiencies \cite{AdenierKhrennikov2007}; the account in
terms of setting- and result-dependent detection probabilities is Bigelow's
\cite{Bigelow2009}, whose fitted $\lambda(s,r)$ has the outcome-dependent structure of
Eq.~(\ref{eq:SSDDE}) below, and the attribution to pair-identification and
post-selection procedures is Kupczynski's \cite{Kupczynski2017}. Reference~\cite{AdenierKhrennikov2003} itself
proposes a test with a controlled source rather than the passive check
considered here. I could locate no sub-percent bound on the $\cos4(a-b)$
content of the fourfold sum; the protocol of Sec.~\ref{sec:protocol} is designed to supply
one.

\section{A protocol for the \texorpdfstring{$\cos 4(a-b)$}{cos 4(a-b)} pixel}\label{sec:protocol}

\paragraph{Two designs.}
Theorem~\ref{thm:hide} follows from $R_d$ depending on $|a-b|$ with period $90^\circ$; it
makes the CHSH angles blind, not every four-point design. Relative angles
$0^\circ$ and $45^\circ$ see the full peak-to-peak, and the $2\times2$
cross-ratio (Design~A)
\begin{equation}
X \;=\; \frac{R_d(0,0)\,R_d(45^\circ,45^\circ)}{R_d(0,45^\circ)\,R_d(45^\circ,0)}
\label{eq:cross}
\end{equation}
cancels every factorizable apparatus response and equals $1.278$ for the
Sec.~\ref{sec:mechanics} profile. For $5\sigma$ at $0.1\%$ peak-to-peak it needs
about $2.5\times10^7$ fourfold coincidences per setting, $10^8$ in all
--- under an hour at $5\times10^4$ coincidences per second,
about a day at $10^3$. Any entanglement-based BBM92 apparatus with passive
basis choice and four detectors per side records this table in normal
operation, so a sub-percent bound may already exist in archived data. The
second design is the full settings torus (Design~B), e.g.\ $24\times24$ points,
Fourier-analyzed as $\log R_d(a,b)$: a factorizable response
$\eta_A(a)\eta_B(b)$ is additive after the logarithm and occupies only the
Fourier axes, while the physical signal occupies the anti-diagonal pixels
$(k,-k)$. The torus resolves the harmonic content at the cost of about
$2.3\times10^8$ coincidences for $5\sigma$ at $0.1\%$ (about
$1.3$ hours at $5\times10^4$ coincidences per second), but only for fringes
large enough to populate the $(4,-4)$ pixel: at $0.1\%$ peak-to-peak the
$\cos8\Delta$ component of the Sec.~\ref{sec:mechanics} profile is $0.0085\%$, below the
per-pixel noise of about $0.009\%$ at that count, so the $5.87{:}1$ shape
is confirmable only for fringes of roughly $0.5\%$ or more. The natural
order is therefore Design~A first, and Design~B with more counts if
Design~A sees something. The four corner points of a torus do not serve as a
$2\times2$, since they carry $1/144$ of the counts.

\paragraph{The non-factorizable apparatus term.}
Not every local apparatus response factorizes. For two-channel analyzers
with outcome-dependent channel efficiencies $\eta_\pm$ on each side and a
maximally entangled source,
\begin{equation}
R_d(a,b)\;=\;\tfrac14\big[S_A(a)S_B(b)\;+\;D_A(a)D_B(b)\,E(a,b)\big],\qquad
S=\eta_++\eta_-,\quad D=\eta_+-\eta_-,
\label{eq:SSDDE}
\end{equation}
and the second term is local in origin but not a product. A constant
relative imbalance $\delta=D/S$ adds no $\cos4(a-b)$ to $R_d$ itself
($R_d\propto 1-V\delta_A\delta_B\cos2\Delta$) but reaches the $(2,-2)$ pixel
of $\log R_d$ at second order, with amplitude $-V^2\delta_A^2\delta_B^2/4$
--- $4\times10^{-4}$ for $20\%$ imbalance per side, comparable to the
$5\times10^{-4}$ target and of opposite sign to the physical fringe, so it
can mask a real fringe as well as mimic one. An angle-dependent imbalance
$d_1\cos(2a-p)$ enters $R_d$ at first order with amplitude $-Vd_1e_1/4$ when
the phases agree ($2.3\times10^{-4}$ at $3\%$ per side), and with equal
phases the quadrature check does not see it: in simulation with $20\%$
constant and $3\%$ angle-dependent imbalance per side at equal phase, the
uncalibrated estimator returns $(-0.064\pm0.006)\%$, larger than the
target amplitude, with a null quadrature. Such variations are real:
Aspect et al.\ report a $\pm1\%$ variation of the detection efficiencies on
rotating the polarimeters, traced to beam displacement on the photocathode
\cite[p.~93]{Aspect1982}. The same letter gives the equality of the two
channel efficiencies of each polarizer ($D=0$ in Eq.~(\ref{eq:SSDDE})) as
what makes its fair-sampling assumption reasonable \cite[p.~92]{Aspect1982};
the calibration here measures that quantity instead. The remedy uses the structure of Eq.~(\ref{eq:SSDDE}), and
reads the imbalances from the coincidences alone. For a maximally entangled source the local
states are maximally mixed, so the four coincidence channels at one setting
are $n_{ij}\propto(1+i\,\delta_A)(1+j\,\delta_B)(1+ij\,E)$ with $i,j=\pm1$,
and the table determines all three quantities in closed form for any $|E|<1$:
\begin{equation}
\ln\frac{1+\delta_A}{1-\delta_A}=\frac12\ln\frac{n_{++}\,n_{+-}}{n_{--}\,n_{-+}},\qquad
\ln\frac{1+\delta_B}{1-\delta_B}=\frac12\ln\frac{n_{++}\,n_{-+}}{n_{--}\,n_{+-}},\qquad
\ln\frac{1+E}{1-E}=\frac12\ln\frac{n_{++}\,n_{--}}{n_{+-}\,n_{-+}}.
\label{eq:closed}
\end{equation}
On the torus $\delta_A$ depends only on $a$ and $\delta_B$ only on $b$, so
the per-setting values are averaged over the other angle with
inverse-variance weights; in the two-setting design the two cross-basis
tables, where $E\approx0$ and all four cells are well populated, supply all
four imbalances. $R_d/(1+\delta_A\delta_B E)$ is then factorizable
response times the physical fringe, to first order in the analyzers' extinction.
Accidental coincidences are subtracted before Eq.~(\ref{eq:closed}) is
applied, from a delayed coincidence window or from the product of the two
detectors' singles rates, background included, and the window; any
residue formed by signal singles carries the same imbalances and keeps the
product form, so it only dilutes $E$ and the fringe.

The two designs demand different accuracy from this calibration. On the torus a
constant imbalance reaches the pixel only through the square of
$\delta_A\delta_B$, as above. In the two-setting design it enters linearly,
$X-1\approx-2V\delta_A\delta_B$, which is $-0.08$ for $20\%$ per side,
forty times the target, so at that imbalance the product
$\delta_A\delta_B$ must be known to a few tenths of a percent of itself. The channel singles give
a second estimate of $\delta_A(a)$ and $\delta_B(b)$, and it is the
cross-check, not the calibration. Background common to the two detectors of
an analyzer (dark counts, stray light), at a fraction $\beta$ of the signal
singles, shrinks the singles estimate to $\delta/(1+\beta)$ and leaves the
imbalance term under-corrected; agreement of the two estimates is therefore
a test that the singles are free of such background and that singles and
coincidence photons see the same channel efficiencies. In simulation with
$20\%$ constant and $3\%$ angle-dependent imbalance per side and no fringe,
$\beta=1\%$ moves the singles-calibrated cross-ratio of the two-setting
design, at $2.5\times10^7$ coincidences per setting, to
$X-1=(-1.7\pm0.4)\times10^{-3}$, against $+2.0\times10^{-3}$ for the
target, and $\beta=100\%$ moves the singles-calibrated torus estimator to
$(-0.031\pm0.012)\%$, against the $0.05\%$ target; the coincidence
calibration returns $(-0.1\pm0.4)\times10^{-3}$ and $(0.001\pm0.005)\%$ in
the same two cases. With $3\%$ factorizable ripple, $20\%$
constant imbalance and $3\%$ angle-dependent imbalance per side at $10^6$
coincidences per point, the uncalibrated estimator returns a null of
$(-0.041\pm0.017)\%$ and the coincidence-calibrated estimator
$(0.000\pm0.006)\%$, linear from $0.05\%$ to the full profile; at equal
phase the calibrated null is $(0.001\pm0.005)\%$. The quoted uncertainties
are standard deviations over simulated runs; where the phases of the
angle-dependent terms are drawn at random, the uncertainties of the
uncalibrated and the singles-calibrated torus estimators are set by that
draw, not by counting noise.

\paragraph{What quantum mechanics predicts.}
For a single pair in any two-qubit state and analyzers built from one
rotating element with fixed imperfections, every POVM element contains only
harmonics $e^{ika}$ with $|k|\le2$ in the analyzer angle, so no coincidence
rate --- single-channel or fourfold --- contains $\cos4(a-b)$. Within
quantum optics, $\cos4(a-b)$ content therefore requires either a
rotation-dependent apparatus response (the angle-dependent imbalance above,
which breaks the fixed-imperfection premise and is calibrated as described)
or more than one pair per coincidence window: two pairs with threshold
detectors give $\mathbb{E}[N_AN_B]=19/8+\cos4\Delta/8$, a term of relative
size roughly $R\tau/8$ ($10^{-4}$ at $10^6$ coincidences per second and a
$1\,{\rm ns}$ window; negligible at $5\times10^4$), with the same sign as the
Sec.~\ref{sec:mechanics} fringe. Repeating the measurement at two pump powers separates it.

\paragraph{Controls.}
Randomized visiting order and heralded-singles normalization against source
drift; interleaved reference points against slow systematics; detectors in
the linear regime; the coincidence-based imbalance calibration with the
channel singles as its cross-check; the quadrature of the $(2,-2)$ pixel as a
partial systematic check; two pump powers against multi-pair emission; and a
source run at or near maximal entanglement, since the fringe prediction is
for the singlet, the harmonic bound above holds for any state, and the
calibration assumes maximally mixed local states: a residual local
polarization --- e.g., at a type-II source, the collection-mode mismatch at the
intersection points of the two emission cones, where the edge of each
collection region is polarization-biased in the source frame --- would
read as an angle-dependent imbalance with its phase set
by the source axes, so it should be bounded independently, e.g.\ by
tomography of the source.

\paragraph{Outcomes.}
Flatness at $0.1\%$ excludes the Sec.~\ref{sec:mechanics} profile and its mixtures with flat
detection to a detection share of order $1\%$, an order of magnitude beyond
the existing bound, and leaves quantum mechanics, flat-rate detection models
of Larsson's type, and measurement-dependent accounts indistinguishable by
rates. It does not force measurement dependence, for the reason given in
Sec.~\ref{sec:data} and one more: at the detection efficiencies of a standard
(non-loophole-free) apparatus, Eq.~\eqref{eq:surface} itself gives $\mathcal{M}=0$, and a
faithful account at the observed $\eta_{\rm eff}$ need pay nothing on either
channel; the forced-dependence statement holds only on the interpolating
family of Sec.~\ref{sec:mpl} under the assumption that the characterized losses are fair,
which is what the measurement tests. A fringe of the predicted shape surviving
the calibrations above would be a joint $(a,b)$ dependence at the $(2,-2)$
pixel that no single-pair, fixed-imperfection apparatus produces, consistent
with a detection account of the Sec.~\ref{sec:mechanics} class at that
mixture level. A fringe of a different shape surviving the same controls
would be the same kind of joint dependence and would exclude this detection
sector as constructed.

\begin{figure}[t]\centering
\includegraphics[width=.62\linewidth]{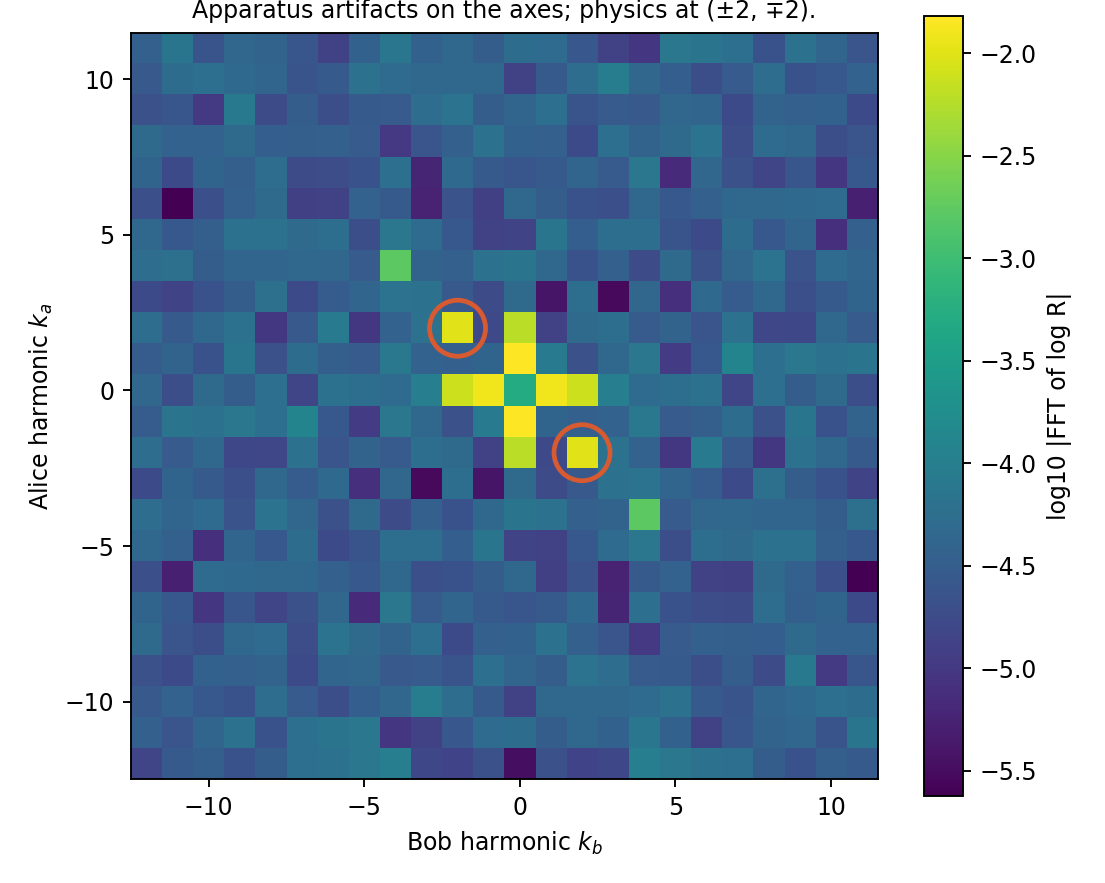}
\caption{Simulated protocol analysis: log-Fourier plane of the rate map
over the settings torus. Factorizable apparatus response occupies the
axes; the physical signal occupies the circled anti-diagonal pixels. Only
factorizable apparatus response is simulated here; the non-factorizable
term of Eq.~(\ref{eq:SSDDE}) is removed by the calibration of
Sec.~\ref{sec:protocol} before this analysis.}
\label{fig:proto}\end{figure}
\section{Discussion and limitations}
The loophole-free experiments \cite{Hensen2015,Giustina2015,Shalm2015}
exclude local models conditional on measurement independence; the cosmic
and human-choice tests \cite{Handsteiner2017,BigBell2018} constrain the
provenance of any MI violation, not its existence, and supermeasured
accounts \cite{HHP2022,GarzaHance2025} locate the dependence in the measure
rather than in a causal link. The surface
\eqref{eq:surface} prices what survives. Limitations: the fringe is a property of the detection class of
Sec.~\ref{sec:mechanics}, and flat-rate faithful detection models outside
that class exist for $\eta_{\rm eff}\le 0.778$ (Sec.~\ref{sec:data}), so
the experiment of Sec.~\ref{sec:protocol} bounds a family rather than
closing a channel; the guilt
trade-off is computed along a specific detection family (an upper bound
on the true detection share at given $\delta_{f}$); the surface is
certified at one point and machine-verified elsewhere, pending a
parametric dual; the quenched ensemble is sampled, the
prices converge to the floor only at sufficient offset coverage (1024
atoms at the registers used, up to 161 pinned angles), and off-register
fidelity is sampled, not bounded; $\rho$ is optimizer-chosen and must be
nonuniform over offsets --- a uniform offset density yields zero
measurement dependence by translation symmetry --- and, although
Sec.~\ref{sec:mpl} fixes its low-order structure at the floor on the
paper's terrain and excludes station-local origins on the terrains
tested, its derivation
from a formation process remains open --- a formation process that stamps
the offsets smoothly pays above Hall's floor by the amounts computed
there, and no census found that stamps them without regard to the flight time
and coupling scale attains it. The specification therefore has two axes: a
census of the offsets grained to $5$--$6^\circ$, and a birth phase
correlated with the approach phases that can stand in for part of that
grain (the premium over the floor at $K=6$ falls $0.70\to0.37\to0.23$
with one and two harmonics; at $K=8$, $0.43\to0.16$), while the detection
rule is fixed by the floor to a sharp step up to a $6^\circ$ undecided
band; a formation law must be written on both axes; the smooth factor $h$ lacks a closed form; and the
protocol's systematics beyond those simulated require the usual
experimental care. The detection structure investigated here was
originally suggested by a beat-pattern (moir\'e) heuristic; none of the
results depend on it.

\section*{Code and data availability}
All linear programs, exact certificates, spectral solvers, and protocol
simulations are provided in the repository directory\\
\url{https://github.com/RandomInternetPreson/moire-phase-space-sampler/tree/main/DST_Bell_MI/Exact_Price_Source_Material},\\
with every numerical claim reproducible from a single script per
section. For Sec.~\ref{sec:mpl} the pipeline is
deterministic end to end (seedless $T=0$ dynamics --- offset draws in
\texttt{dither\_sweep.py} use recorded seeds; independent
re-executions reproduce every price bit-for-bit) and runs on
commodity hardware:
\texttt{breed\_t0.py} selects the interaction weights under the
faithfulness objective; \texttt{flight\_time.py} (with
\texttt{NANG=41} and the tolerance sweep \texttt{EPS}) produces the
uniform-register price table; \texttt{extract\_pattern.py} re-solves
at $\varepsilon=0.01$ and emits the fingerprint of
Fig.~\ref{fig:pattern} (\texttt{pattern.json}); and
\texttt{orchestra.py} performs the detection-composition scan.
\texttt{slice\_solve2.py} re-solves the source register at denser
melodies (81, 161 angles) with an optional second constrained slice;
\texttt{slice\_check\_mp.py} evaluates any saved solution on arbitrary
slices and angle windows (emitting \texttt{slice\_check\_a0\_*.json});
\texttt{dither\_sweep\_rho.py} extends \texttt{dither\_sweep.py} with
trajectory-bank checkpointing and solution export, producing the
1024-offset, 161-angle floor solution
(\texttt{rho\_nsh32\_random0\_eps0.01\_nang161.json}).
\texttt{dither\_sweep.py} performs the offset-ensemble robustness
battery (lattice, shifted lattice, stratified, random; 256 and 1024
atoms), emitting \texttt{dither\_stage1/2/3.json}. Total wall time on a 56-core workstation is under two hours for the
256-atom pipeline; each 1024-atom run takes about one hour. The
$\varepsilon=0.01$ price is cross-validated by two independently
coded solvers on the same register (\texttt{extract\_pattern.py}
and \texttt{flight\_time.py} with its constraint grid set to
$[2^\circ,178^\circ]$), both returning $1.375$ to three decimals,
and the $\varepsilon=0.02$ price is independently recovered by
\texttt{orchestra.py} at $\alpha=0$ ($\md=0.17048$ in both). For the
census structure of Sec.~\ref{sec:mpl}, \texttt{census\_structure\_v2.py}
computes the Fourier skeleton and permutation null of a saved solution,
\texttt{terrain\_factor\_v1.py} tests whether $U$ factors as two screens
(it does not), \texttt{flight\_time\_perturb\_v1.py} and
\texttt{dither\_sweep\_perturb\_v1.py} re-solve on a perturbed terrain,
\texttt{product\_census\_v3.py} tests station-local censuses on the saved
banks (\texttt{product\_census\_v4.py} adds sum/difference coordinates), and
\texttt{breed\_product\_v1.py} with \texttt{pipeline\_v1.py} breeds and
validates a terrain scored by the station-local price;
\texttt{census\_bandlimit\_v4.py}, run through
\texttt{run\_parallel\_v1.py}--\texttt{v8.py}, computes the band-limited
prices and the $|k|\le5$ certificates, \texttt{pipeline\_4096\_v2.py}
builds and solves the 4096-offset and lattice censuses,
\texttt{dither\_sweep\_swap\_v1.py} performs the roles-swapped solve,
\texttt{census\_cells\_v1.py} measures the cell dependence and
\texttt{census\_cellsep\_v1.py}/\texttt{v2.py} price the cell-blind
classes, \texttt{pipeline\_bulge\_v1.py} with
\texttt{census\_profile\_v1.py}--\texttt{v4.py} and
\texttt{run\_bulge2\_v1.py} stores the phase histograms and prices
the detection profiles, and the joint births are solved and inspected by
\texttt{pipeline\_fuzz\_v1.py} with
\texttt{census\_fuzz\_v1.py}--\texttt{v3.py} and
\texttt{rim\_skeleton\_v1.py}; the
solutions, banks and transcripts are deposited at
\url{https://doi.org/10.5281/zenodo.23537610} (version 6). For
Secs.~\ref{sec:fringe}--\ref{sec:protocol}, \texttt{solve\_h\_spectral.py}
produces the exact-profile harmonics and \texttt{protocol\_sensitivity\_v2\_3.py}
the protocol sensitivities and the imbalance nulls (including the
equal-phase case and background in the singles) under both calibrations; for
data, \texttt{timetag\_to\_table\_v1.py} counts the coincidence tables from
raw time-tags and writes them with the detector singles, the run duration
and the window, and \texttt{cross\_ratio\_v2.py} subtracts accidentals from
those and returns the calibrated cross-ratio of Eq.~(\ref{eq:cross}).

\section*{Acknowledgments}
I thank Jonte Hance (Newcastle University) for a reproduction of the price surface
and the detection profile, and for the review that led to the corrections
of Secs.~\ref{sec:data} and~\ref{sec:protocol}: the Aspect and
Brida--Genovese--Piacentini bounds, the Larsson counterexample and the
symmetrized Gisin--Gisin model, the two-channel imbalance term and the correction that removes it, the
single-pair harmonic argument, the multi-pair estimate, and the
cross-ratio design.

\emph{In memoriam.} Michael J.~W. Hall, whose bound is the anchor of this
paper, died in April 2026. This work is dedicated to his memory.

\section*{AI-assisted preparation}
Large-language-model assistance (Claude, Anthropic) was used during
this research as a tool for checking derivations, drafting text, and
writing verification code, under the sole direction of the author, who
takes full responsibility for the entire content. All exact values and
identities were verified independently of the tool by machine
computation as described above, and every numerical claim is
reproducible from the released scripts.

\end{document}